\documentclass[11pt]{article}
\usepackage{graphicx,amssymb,amsmath}

\usepackage[linesnumbered, ruled]{algorithm2e}
\SetKwInput{KwResult}{Function}

\usepackage{amsfonts}
\usepackage{amsmath}
\usepackage{booktabs}
\usepackage{cite}
\usepackage{color}
\usepackage{enumerate}
\usepackage{float}
\usepackage{hyperref}
\usepackage{mathrsfs}
\usepackage{subfig}
\usepackage{tabularx}
\usepackage{fullpage}
\usepackage[top=0.8in, bottom=0.9in, left=0.9in, right=0.9in]{geometry}

\graphicspath{{./images/}}

\usepackage{lineno}

\def\calC{\mathcal{C}}

\def\calD{\mathcal{D}}

\def\calG{\mathcal{G}}
\def\calL{\mathcal{L}}

\def\calP{\mathcal{P}}
\def\calE{\mathcal{E}}

\newtheorem{lemma}{Lemma}
\newtheorem{theorem}{Theorem}
\newtheorem{problem}{Problem}
\newtheorem{corollary}{Corollary}

\newenvironment{proof}{\noindent {\textbf{Proof:}}\rm}{\hfill $\Box$
\rm\bigskip}

\title{Improved Algorithms for Distance Selection and Related Problems\thanks{A preliminary version of this paper appeared in {\em Proceedings of the 31st Annual European Symposium on Algorithms (ESA 2023)}. This research was supported in part by NSF under Grants CCF-2005323 and CCF-2300356. }}

\author{
Haitao Wang\thanks{Kahlert School of Computing,
University of Utah, Salt Lake City, UT 84112, USA. {\tt haitao.wang@utah.edu}}
\and
Yiming Zhao\thanks{Corresponding author. Department of Computer Sciences,
Metropolitan State University of Denver, Denver, CO 80217, USA. {\tt yizhao@msudenver.edu}}
}

\begin{document}
\pagestyle{plain}
\date{}

\thispagestyle{empty}
\maketitle

\begin{abstract}
In this paper, we propose new techniques for solving geometric optimization problems involving interpoint distances of a point set in the plane. Given a set $P$ of $n$ points in the plane and an integer $1 \leq k \leq \binom{n}{2}$, the distance selection problem is to find the $k$-th smallest interpoint distance among all pairs of points of $P$. The previously best deterministic algorithm solves the problem in $O(n^{4/3} \log^2 n)$ time [Katz and Sharir, SIAM J. Comput. 1997 and SoCG 1993]. In this paper, we improve their algorithm to $O(n^{4/3} \log n)$ time. Using similar techniques, we also give improved algorithms on both the two-sided and the one-sided discrete Fr\'{e}chet distance with shortcuts problem for two point sets in the plane. For the two-sided problem (resp., one-sided problem), we improve the previous work [Avraham, Filtser, Kaplan, Katz, and Sharir, ACM Trans. Algorithms 2015 and SoCG 2014] by a factor of roughly $\log^2(m+n)$ (resp., $(m+n)^{\epsilon}$), where $m$ and $n$ are the sizes of the two input point sets, respectively. Other problems whose solutions can be improved by our techniques include the reverse shortest path problems for unit-disk graphs. Our techniques are quite general and we believe they will find many other applications in future.
\end{abstract}

\section{Introduction}
\label{sec:Introduction}

In this paper, we propose new techniques for solving geometric optimization problems involving interpoint distances in a point set in the plane. More specifically, the optimal objective value of these problems is equal to the (Euclidean) distance of two points in the set. Our techniques usually yield improvements over the previous work by at least a logarithmic factor (and sometimes a polynomial factor).

The first problem we consider is the \emph{distance selection} problem: Given a set $P$ of $n$ points in the plane and an integer $1 \leq k \leq \binom{n}{2}$, the problem asks for the $k$-th smallest interpoint distance among all pairs of points of $P$.
The problem can be easily solved in $O(n^2)$ time.
The first subquadratic time algorithm was given by Chazelle~\cite{ref:ChazelleNe85}; the algorithm runs in $O(n^{9/5} \log^{4/5} n)$ time and is based on Yao's technique~\cite{ref:YaoOn82}. 
Later, Agarwal, Aronov, Sharir, and Suri~\cite{ref:AgarwalSe93} gave a better algorithm of $O(n^{3/2} \log^{5/2} n)$ time and subsequently Goodrich~\cite{ref:GoodrichGe93} solved the problem in $O(n^{4/3}\log^{8/3}n)$ time. Katz and Sharir~\cite{ref:KatzAn97} finally presented an $O(n^{4/3}\log^2 n)$ time algorithm. All above are deterministic algorithms. Several randomized algorithms have also been proposed for the problem. The randomized algorithm of \cite{ref:AgarwalSe93} runs in $O(n^{4/3} \log^{8/3} n)$ expected time. Matousek~\cite{ref:MatousekRa91} gave another randomized algorithm of $O(n^{4/3} \log^{2/3} n)$ expected time. Very recently, Chan and Zheng proposed a randomized algorithm of $O(n^{4/3})$ expected time (see the arXiv version of~\cite{ref:ChanHo22}). Also, the time complexity can be made as a function of $k$. In particular, Chan's randomized techniques~\cite{ref:ChanOn01} solved the problem in $O(n\log n+n^{2/3}k^{1/3}\log^{5/3}n)$ expected time and Wang~\cite{ref:WangUn23} recently improved the algorithm to $O(n\log n+n^{2/3}k^{1/3}\log n)$ expected time; these algorithms are particularly interesting when $k$ is relatively small.

In this paper, we present a new deterministic algorithm that solves the distance selection problem in $O(n^{4/3}\log n)$ time. Albeit slower than the randomized algorithm of Chan and Zheng~\cite{ref:ChanHo22}, our algorithm is the first progress on the deterministic solution since the work of Katz and Sharir~\cite{ref:KatzAn97} published 25 years ago (30 years if we consider their conference version in SoCG 1993).

One technique we introduce is an algorithm for solving the following {\em partial batched range searching problem}. 
\begin{problem}
    \label{prob:BRS}
    {\em \bf (Partial batched range searching)}
    Given a set $A$ of $m$ points and a set $B$ of $n$ points in the plane and an interval $(\alpha, \beta]$, one needs to construct two collections of edge-disjoint complete bipartite graphs $\Gamma(A, B, \alpha, \beta) = \{A_t \times B_t\ |\ A_t \subseteq A, B_t \subseteq B\}$ and $\Pi(A, B, \alpha, \beta) = \{A'_s \times B'_s\ |\ A'_s \subseteq A, B'_s \subseteq B\}$ such that the following two conditions are satisfied (see Fig.~\ref{fig:PBRS} for an example):
    \begin{enumerate}
        \item For each pair $(a, b) \in A_t \times B_t\in \Gamma(A, B, \alpha, \beta)$, the (Euclidean) distance $\lVert ab \rVert$ between points $a \in A_t$ and $b \in B_t$ is in $(\alpha, \beta]$.
        \item For any two points $a \in A$ and $b \in B$ with $\lVert ab \rVert \in (\alpha, \beta]$, either $\Gamma(A, B, \alpha, \beta)$ has a unique graph  $A_t \times B_t$ that contains $(a, b)$ or $\Pi(A, B, \alpha, \beta)$ has a unique graph  $A'_s \times B'_s$ that contains $(a, b)$.
    \end{enumerate}
In other words, the two collections $\Gamma$ and $\Pi$ together record all pairs $(a,b)$ of points $a\in A$ and $b\in B$ whose distances are in $(\alpha,\beta]$. While all pairs of points recorded in $\Gamma$ have their distances in $(\alpha,\beta]$, this may not be true for $\Pi$. For this reason, we sometimes call the point pairs recorded in $\Pi$ {\em uncertain pairs}. 
\end{problem}

\begin{figure}
    \centering
    \includegraphics[height=1.8in]{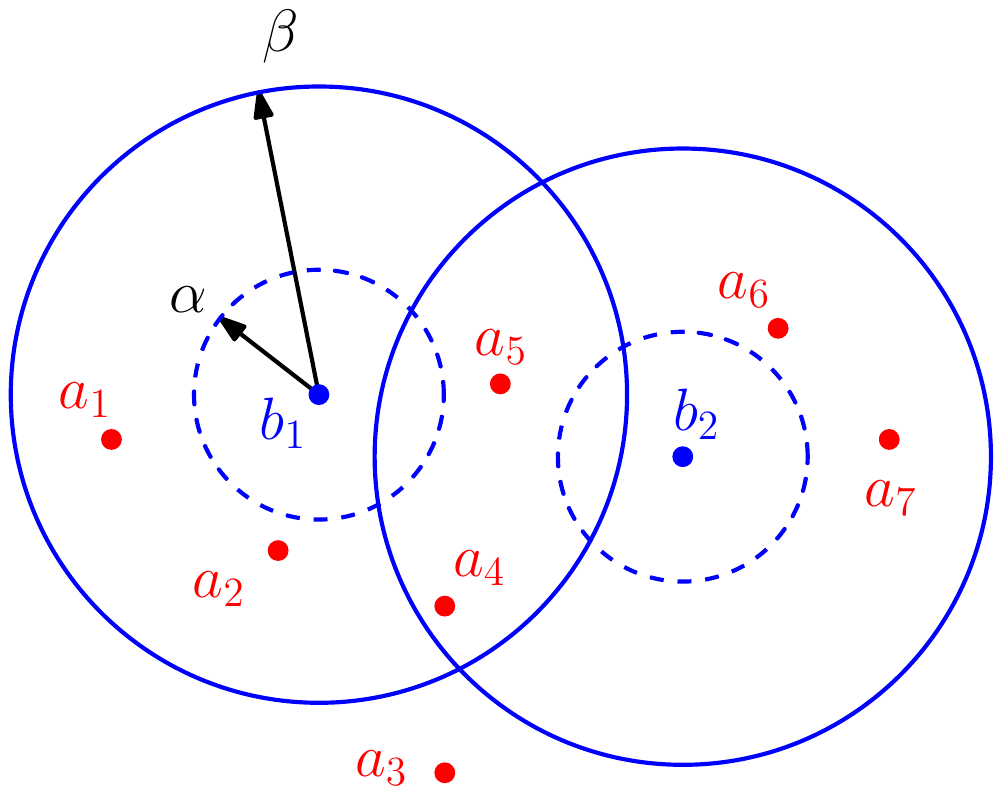}
    \caption{An example of Problem~\ref{prob:BRS} with $\Gamma = \{ \{a_1, a_2\} \times \{b_1\}, \{a_4, a_5\} \times \{b_1, b_2\}\}$ and $\Pi = \{\{a_3, a_6, a_7\} \times \{b_2\}\}$. Note that the uncertain pair $(a_3, b_2)$ has a distance $\lVert a_3 b_2 \rVert \notin (\alpha, \beta]$.} 
    \label{fig:PBRS}
\end{figure}

Note that if context is clear, we sometimes use $\Gamma$ and $\Pi$ to refer to $\Gamma(A, B, \alpha, \beta)$ and $\Pi(A, B, \alpha, \beta)$, respectively. Also, for short, we use BRS to refer to batched range searching. 


In the traditional BRS, which has been studied with many applications, e.g.,\cite{ref:WangRe23,ref:KatzEf21arXiv, ref:AvrahamTh15}, the collection $\Pi$ is $\emptyset$ (and thus $\Gamma$ itself satisfies the two conditions in Problem~\ref{prob:BRS}); for differentiation, we refer to this case as the {\em complete BRS}.
The advantage of the partial problem over the complete problem is that the partial problem can usually be solved faster, with a sacrifice that some uncertain pairs (i.e., those recorded in $\Pi$) are left unresolved. As will be seen later, in typical applications the number of those uncertain pairs can be made small enough so that they can be handled easily without affecting the overall runtime of the algorithm. More specifically, we derive an algorithm to compute a solution for the partial BRS, whose runtime is controlled by a parameter (roughly speaking, the runtime increases as the graph sizes of $\Pi$ decreases). Previously, Katz and Sharir~\cite{ref:KatzAn97} gave an algorithm for the complete problem. Our solution, albeit for the more general partial problem, even improves their algorithm by roughly a logarithmic factor when applied to the complete case.

On the one hand, our partial BRS solution helps achieve our new  result for the distance selection problem. On the other hand, combining some techniques for the latter problem, we propose a general algorithmic framework that can be used to solve any geometric optimization problem that involves interpoint distances of a set of points in the plane.
Consider such a problem whose optimal objective value (denoted by $\delta^*$) is equal to the distance of two points of a set $P$ of $n$ points in the plane. Assume that the decision problem (i.e., given $\delta$, decide whether $\delta\geq \delta^*$) can be solved in $T_D$ time. A straightforward algorithm for computing $\delta^*$ is to use the distance selection algorithm and the decision algorithm to perform binary search on interpoint distances of all pairs of points of $P$; the algorithm runs in $O(\log n)$ iterations and each iteration takes $O(n^{4/3} \log n + T_D)$ time (if we use our new distance selection algorithm). As such, the total runtime is $O((n^{4/3} \log n + T_D) \log n)$. Using our new framework, the runtime can be bounded by $O((n^{4/3} + T_D) \log n)$, which is faster when $T_D=o(n^{4/3}\log n)$.

\paragraph{Two-sided DFD.}
One application of this new framework is the \emph{two-sided discrete Fr\'{e}chet distance with shortcuts} problem, or {\em two-sided DFD} for short. Fr\'{e}chet distance is used to measure the similarity between two curves and many of its variations have been studied, e.g., \cite{ref:AgarwalCo14-DFD,ref:AltCo95,ref:AvrahamTh15,ref:BuchinEx09, ref:BuchinCo14,ref:DriemelJa13}. To reduce the impact of outliers between two (sampled) curves, discrete Fr\'{e}chet distance with shortcuts was proposed~\cite{ref:AvrahamTh15,ref:DriemelJa13}. If outliers of only one curve need to be taken care of, it is called {\em one-sided DFD}; otherwise it is {\em two-sided DFD}.
Avraham, Filtser, Kaplan, Katz, and Sharir~\cite{ref:AvrahamTh15} solved the two-sided DFD in $O((m^{2/3} n^{2/3} + m + n) \log^3 (m + n))$, where $m$ and $n$ are the numbers of vertices of the two input curves, respectively. Using our new framework, we improve their algorithm to $O((m^{2/3} n^{2/3} \cdot 2^{O(\log^* (m + n))} + m \log n + n \log m) \log (m + n))$ time, an improvement of roughly $O(\log^2 (m+n))$.

\paragraph{One-sided DFD.}
For the one-sided DFD, the authors of \cite{ref:AvrahamTh15} gave a randomized algorithm of $O((m + n)^{6/5 + \epsilon})$ expected time, for any constant $\epsilon > 0$. Using our solution to the partial BRS, we improve their algorithm to $O((m + n)^{6/5} \log^{7/5} (m + n))$ expected time. Combining the interval shrinking and the bifurcation tree techniques~\cite{ref:AvrahamTh15,ref:KaplanTh23}, our partial BRS results lead to an algorithmic framework for solving geometric optimization problems that involve interpoint distances in a point set in the plane.
Consider such a problem whose optimal objective value (denoted by $\delta^*$) is equal to the distance of two points of a set $P$ of $n$ points in the plane. The framework has two main procedures. The first procedure is to compute an interval that contains $\delta^*$ and with high probability at most $L$ interpoint distances of $P$. Using the interval and a \emph{bifurcation tree} technique, the second main procedure finally computes $\delta^*$.
Assuming that the decision problem can be solved in $T_D$ time, the first main procedure takes $O(n^{4/3} / L^{1/3} \cdot \log^2 n + T_D \cdot \log n \cdot \log \log n)$ expected time. Assuming that the decision problem can be solved in $T^*_D$ time by a special algorithm in which all critical values are interpoint distances (see Section~\ref{sec:OS-DFD} for more details), 
the second main procedure runs in $O(\sqrt{L\cdot T_D\cdot T^*_D\cdot \log  n})$ expected time. As such, the total expected time of the entire algorithm is $O(n^{4/3} / L^{1/3} \cdot \log^2 n + T_D \cdot \log n \cdot \log \log n + \sqrt{L\cdot T_D\cdot T^*_D\cdot \log  n})$. Our result for the one-sided DFD is a direct application of this framework. More specifically, since both $T_D,T_D^*=O(m+n)$~\cite{ref:AvrahamTh15}, we set $L=(m+n)^{2/5}\log^{9/5}(m+n)$ and replace $n$ by $(m+n)$ in the above time complexity as there are two parameters $m$ and $n$ in the problem.

\paragraph{Reverse shortest paths in unit-disk graphs.}
We demonstrate two more applications of the framework where our new techniques lead to improved results over the previous work: the \emph{reverse shortest paths in unit-disk graphs} and its weighted case.
Given a set $P$ of $n$ points in the plane and a parameter $\delta > 0$, the unit-disk graph $G_{\delta}(P)$ is an undirected graph whose vertex set is $P$ such that an edge connects two points $p, q \in P$ if the (Euclidean) distance between $p$ and $q$ is at most $\delta$. In the unweighted (resp., weighted) case, the weight of each edge is equal to $1$ (resp., the distance between the two vertices). Given set $P$, two points $s, t \in P$, and a parameter $\lambda$, the problem is to compute the smallest $\delta^*$ such that the shortest path length between $s$ and $t$ in $G_{\delta^*}(P)$ is at most $\lambda$.

Deterministic algorithms of $O(n^{5/4} \log^{7/4}n)$ and $O(n^{5/4} \log^{5/2}n)$ times are known for the unweighted and weighted problems, respectively~\cite{ref:WangRe23}. Kaplan, Katz, Saban, and Sharir~\cite{ref:KaplanTh23} recently gave randomized algorithms that solve them in $O^*(n^{6/5})$ expected time, where the notation $O^*$ hides a subpolynomial factor. Our new framework leads to new randomized algorithms of $O(n^{6/5} \log^{8/5} n)$ and $O(n^{6/5} \log^{11/5} n)$ expected time, respectively. Comparing to \cite{ref:KaplanTh23}, our effort is to make the subpolynomial factor small.

\paragraph{Recap.}
In summary, we propose two algorithmic frameworks for solving geometric optimization problems that involve interpoint distances in a set of points in the plane. The first fra1mework is deterministic while the second one is randomized. The first framework is mainly useful when the decision algorithm time $T_D$ is relatively large (e.g., close to $O(n^{4/3})$) while the second one is more interesting when it is possible to make $T^*_D$ relatively small (e.g., near linear). Both frameworks rely on our solution to the partial BRS problem. As optimization problems involving interpoint distances are very common in computational geometry, we believe our techniques will find more applications in future.

\paragraph{Outline.} The rest of the paper is organized as follows.
Section~\ref{sec:BRS} presents our algorithm for the partial BRS. The algorithm for the distance selection problem is described in Section~\ref{sec:DistanceSelection}. The two-sided DFD problem is solved in Section~\ref{sec:TS-DFD}, where we also propose our first algorithmic framework. The one-sided DFD problem and our second algorithmic framework as well as the reverse shortest path problem are discussed in Section~\ref{sec:OS-DFD}.

\section{Partial batched range searching}
\label{sec:BRS}

In this section, we present our solution to the partial BRS problem, i.e., Problem~\ref{prob:BRS}. We follow the notation in the statement of Problem~\ref{prob:BRS}. In particular, $m=|A|$ and $n=|B|$.

For any set $P$ of points and a compact region $R$ in the plane, let $P(R)$ denote the subset of points of $P$ in $R$, i.e., $P(R)=P\cap R$. 
For any point $p$ in the plane, with respect to the interval $(\alpha, \beta]$ in Problem~\ref{prob:BRS}, let $D_p$ denote the annulus centered at $p$ and having radii $\alpha$ and $\beta$ (e.g., see Fig.~\ref{fig:annuli}); so $D_p$ has an inner boundary circle of radius $\alpha$ and an outer boundary circle of radius $\beta$. We assume that $D_p$ includes its outer boundary circle but not its inner boundary circle. In this way, a point $q$ is in $D_p$ if and only if $\lVert pq \rVert\in (\alpha,\beta]$. 
Define $\calD$ as the set of all annuli $D_p$ for all points $p\in A$. Define $\calC$ to be the set of boundary circles of all annuli of $\calD$. Hence, $\calC$ consists of $2m$ circles. For any compact region $R$ in the plane, let $\calC_R$ denote the subset of circles of $\calC$ that intersect the relative interior of $R$.

An important tool we use is the cuttings~\cite{ref:ChazelleCu93}. For a parameter $1 \leq r \leq n$, a $(1/r)$-cutting $\Xi$ of size $O(r^2)$ for $\calC$ is a collection of $O(r^2)$ constant-complexity cells whose union covers the plane such that the interior of each cell $\sigma\in \Xi$ is intersected by at most $m / r$ circles in $\calC$, i.e., $|\calC_{\sigma}| \leq m/r$.

We actually use \emph{hierarchical cuttings}~\cite{ref:ChazelleCu93}. We say that a cutting $\Xi'$ \emph{$c$-refines} a cutting $\Xi$ if each cell of $\Xi'$ is contained in a single cell of $\Xi$ and every cell of $\Xi$ contains at most $c$ cells of $\Xi'$. Let $\Xi_0$ denote the cutting whose single cell is the whole plane. Then we define cuttings $\{\Xi_0, \Xi_1, ..., \Xi_k\}$, in which each $\Xi_i$, $1 \leq i \leq k$, is a $(1/\rho^i)$-cutting of size $O(\rho^{2i})$ that $c$-refines $\Xi_{i - 1}$, for two constants $\rho$ and $c$. By setting $k = \lceil \log_\rho r \rceil$, the last cutting $\Xi_k$ is a $(1/r)$-cutting. The sequence $\{\Xi_0, \Xi_1, ..., \Xi_k\}$ of cuttings is called a hierarchical $(1/r)$-cutting of $\calC$. For a cell $\sigma'$ of $\Xi_{i - 1}$, $1 \leq i \leq k$, that fully contains cell $\sigma$ of $\Xi_i$, we say that $\sigma'$ is the \emph{parent} of $\sigma$ and $\sigma$ is a \emph{child} of $\sigma'$. Thus the hierarchical $(1/r)$-cutting can be viewed as a tree structure with $\Xi_0$ as the root.

A hierarchical $(1/r)$-cutting of $\calC$ can be computed in $O(mr)$ time, e.g., by the algorithm in \cite{ref:WangUn23}, which adapts Chazelle's algorithm~\cite{ref:ChazelleCu93} for hyperplanes. The algorithm also produces the subset $\calC_\sigma$ for all cells $\sigma\in \Xi_i$ for all $i=0,1,\ldots,k$, implying that the total size of these subsets is bounded by $O(mr)$. In particular, each cell of the cutting produced by the algorithm of \cite{ref:WangUn23} is a {\em pseudo-trapezoid} that is bounded by two vertical line segments from left and right, an arc of a circle of $\calC$ from top, and an arc of a circle of $\calC$ from bottom (e.g., see Fig.~\ref{fig:pseudotrap}).

\begin{figure}[t]
\begin{minipage}[t]{0.48\textwidth}
\begin{center}
\includegraphics[height=1.3in]{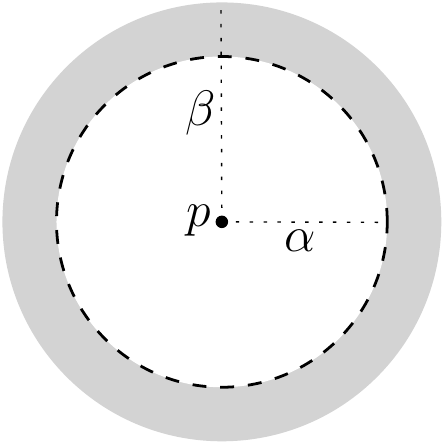}
\caption{\footnotesize Illustrating an annulus $D_p$ (the grey region).}
\label{fig:annuli}
\end{center}
\end{minipage}
\begin{minipage}[t]{0.48\textwidth}
\begin{center}
\includegraphics[height=1.2in]{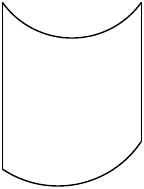}
\caption{\footnotesize Illustrating a pseudo-trapezoid.}
\label{fig:pseudotrap}
\end{center}
\end{minipage}
\end{figure}

Using cuttings, we obtain the following solution to the partial BRS problem. 
\begin{lemma}
    \label{lem:10}
    For any $r$ with $1\leq r\leq \min\{m^{1/3},n^{1/3}\}$, we can compute in $O(mr \log r + nr)$ time two collections $\Gamma(A, B, \alpha, \beta)= \{A_t \times B_t\ |\ A_t \subseteq A, B_t \subseteq B\}$ 
     and $\Pi(A, B, \alpha, \beta) = \{A'_s \times B'_s\ |\ A'_s \subseteq A, B'_s \subseteq B\}$ of edge-disjoint complete bipartite graphs that satisfy the conditions of Problem~\ref{prob:BRS},
     with the following complexities: (1) $|\Gamma|=O(r^4)$; (2) $\sum_t |A_t|, \sum_t |B_t| = O(mr \log r + nr)$; (3) $|\Pi|=O(r^4)$; (4) $|A_s'| = O(m/r^3)$ and $|B_s'| = O(n/r^3)$ for each $A_s' \times B_s' \in \Pi$; (5) the number of pairs of points recorded in $\Pi$ is $O(r^4\cdot m/r^3\cdot n/r^3)=O(mn/r^2)$.
\end{lemma}
\begin{proof}
    \label{proof:lemma-UTBRS}
    We begin with constructing a hierarchical $(1/r)$-cutting $\{\Xi_0, \Xi_1, ..., \Xi_k\}$ for $\calC$, which takes $O(mr)$ time as discussed above. We use $\Xi$ to refer to the set of all cells $\sigma$ in all cuttings $\Xi_i$, $0\leq i\leq k$. 
    Next we compute the set $B(\sigma)$ for each cell $\sigma$ in the cutting (recall that $B(\sigma)$ refers to the subset of points of $B$ inside $\sigma$; we call $B(\sigma)$ a {\em canonical subset}). This can be done in $O(n\log r)$ time in a top-down manner by processing each point of $B$ individually. Specifically, for each point $p \in B$, suppose we know that $p$ is in $\sigma'$ for a cell $\sigma'$ in $\Xi_{i-1}$ (which is true initially when $i=1$ as $\Xi_0$ has a single cell that is the entire plane). By examining each child of $\sigma'$ we can find in $O(1)$ time the cell $\sigma$ of $\Xi_i$ that contains $p$ and then we add $p$ to $B(\sigma)$. Since $k=\Theta(\log r)$, each point of $B$ is stored in $O(\log r)$ canonical subsets and the total size of all canonical subsets $B(\sigma)$ for all cells $\sigma\in \Xi$ is $O(n \log r)$. 

    Next, for each cell $\sigma$ of $\Xi$, we compute another canonical subset $A_{\sigma}\subseteq A$. Specifically, a point $p\in A$ is in $A_{\sigma}$ if the annulus $D_p$ contains $\sigma$ but not $\sigma$'s parent. The subsets $A_{\sigma}$ for all cells $\sigma$ of $\Xi$ can be computed in $O(mr)$ time. Indeed, recall that the cutting algorithm already computes $\calC_{\sigma}$ for all cells $\sigma\in \Xi$. For each $\Xi_{i-1}$, $1\leq i\leq k$, for each cell $\sigma'$ of $\Xi_{i-1}$, we consider each circle $C\in \calC_{\sigma'}$. Let $p$ be the point of $A$ such that $C$ is a bounding circle of the annulus $D_p$. For each child $\sigma$ of $\sigma'$, if $D_p$ fully contains $\sigma$, then we add $p$ to $A_{\sigma}$. In this way, $A_{\sigma}$ for all cells $\sigma$ of $\Xi$ can be computed in $O(mr)$ time since $\sum_{0\leq i\leq k} \sum_{\sigma'\in \Xi_i} |\calC_{\sigma'}|=O(mr)$ and each cell $\sigma'$ has $O(1)$ children. As such, the total size of $A_{\sigma}$ for all cells $\sigma\in \Xi$ is $O(mr)$. 
    
    By definition, for each cell $\sigma\in \Xi$, for any point $a\in A_{\sigma}$ and any point $b\in B(\sigma)$, we have $\lVert ab\rVert\in (\alpha,\beta]$. As such, we return $\{A_{\sigma}\times B(\sigma)\ |\ \sigma\in \Xi\}$ as a subcollection of $\Gamma(A,B,\alpha,\beta)$ to be computed for the lemma. Note that the complete bipartite graphs of $\{A_{\sigma}\times B(\sigma)\ |\ \sigma\in \Xi\}$ are edge-disjoint. 
    The size of the subcollection is equal to the number of cells of the hierarchical cutting, which is $O(r^2)$. Also, we have shown above that $\sum_{\sigma\in \Xi}|A_{\sigma}|=O(mr)$ and $\sum_{\sigma\in \Xi}|B(\sigma)|=O(n\log r)$. 

    For each cell $\sigma$ of the last cutting $\Xi_k$, we have $|\calC_{\sigma}| \leq m / r$. Let $\hat{A}_{\sigma}$ denote the subset of points $p\in A$ such that $D_p$ has a bounding circle in $\calC_{\sigma}$.  
    We do not know whether distances between points of $\hat{A}_{\sigma}$ and points of $B(\sigma)$ are in $(\alpha, \beta]$ or not. If $|B(\sigma)|>n / r^2$, then we arbitrarily partition $B(\sigma)$ into subsets of size between $n / (2r^2)$ and $n / r^2$. We call these subsets \emph{standard subsets} of $B(\sigma)$. Since $|B| = n$ and we have $O(r^2)$ cells in cutting $\Xi_k$, the number of standard subsets of all cells of $\Xi_k$ is $O(r^2)$. For each standard subset $\hat{B}(\sigma)\subseteq B(\sigma)$, we form a pair $(\hat{A}_{\sigma}, \hat{B}(\sigma))$ as an ``unsolved'' {\em subproblem}. 
      Then we have $O(r^2)$ subproblems. 
    Note that $|\hat{A}_{\sigma}|\leq m/r$ and $|\hat{B}(\sigma)|\leq n/r^2$. 
  If we apply the same algorithm recursively on each subproblem, then we have the following recurrence relation (which holds for any $1\leq r\leq m$):
    \begin{gather}
        \label{formula:PrimalRecurrence}
        T(m, n) = O(mr + n \log r) + O(r^2) \cdot T(\frac{m}{r}, \frac{n}{r^2})
    \end{gather}
    
    Note that if we use $T(m,n)$ to represent the total size of $A_t$ and $B_t$ of all complete bipartite graphs $A_t\times B_t$ in the subcollection of $\Gamma(A,B,\alpha,\beta)$ that have been produced as above, then we have the same recurrence as above. If $N(m,n)$ denotes the number of these graphs, then we have the following recurrence: 
     \begin{gather*}
        \label{formula:PrimalRecurrenceN}
        N(m, n) = O(r^2) + O(r^2) \cdot N(\frac{m}{r}, \frac{n}{r^2})
    \end{gather*}

    We now solve the problem in a ``dual'' setting by switching the roles of $A$ and $B$, i.e., define annuli centered at points of $B$ and compute the hierarchical cutting for their bounding circles. Then, symmetrically we have the following recurrences (which holds for any $1\leq r\leq n$):
    \begin{gather}
        \label{formula:DualRecurrence}
        T(m, n) = O(nr + m \log r) + O(r^2) \cdot T(\frac{m}{r^2},\frac{n}{r})
    \end{gather}
     \begin{gather*}
        \label{formula:DualRecurrenceN}
        N(m, n) = O(r^2) + O(r^2) \cdot N(\frac{m}{r^2},\frac{n}{r})
    \end{gather*}

    By applying \eqref{formula:DualRecurrence} to each subproblem of \eqref{formula:PrimalRecurrence} using the same parameter $r$ and we can obtain the following recurrence:
    \begin{gather*}
        \label{formula:balanced}
        T(m, n) = O(mr \log r + nr) + O(r^4) \cdot T(\frac{m}{r^3}, \frac{n}{r^3})
    \end{gather*}
    Similarly, we have
     \begin{gather*}
        \label{formula:balancedN}
        N(m, n) = O(r^4) + O(r^4) \cdot N(\frac{m}{r^3}, \frac{n}{r^3})
    \end{gather*}

    The above recurrences tell us that in $O(mr\log r + nr)$ time we can compute a collection of $O(r^4)$ edge-disjoint complete bipartite graphs $A_t\times B_t$ with $A_t\subseteq A$ and $B_t\subseteq B$ such that for any two points $a\in A_t$ and $b\in B_t$ their distance $\lVert ab\rVert$ lies in $(\alpha,\beta]$. 
    Further, the size of all such $A_t$'s and $B_t$'s is bounded by $O(mr\log r + nr)$. 
    We return the above collection as $\Gamma(A,B,\alpha,\beta)$ for the lemma. 
    
    In addition, we have also $O(r^4)$ graphs $A'_s\times B'_s$ with $A'_s\subseteq A$ and $B'_s\subseteq B$ corresponding to the unsolved subproblems $T(m/r^3,n/r^3)$ and we do not know whether $\lVert ab\rVert\in (\alpha,\beta]$ for points $a\in A'_s$ and $b\in B'_s$. We return the collection of all such graphs as $\Pi(A,B,\alpha,\beta)$ for the lemma. Hence, $|\Pi(A,B,\alpha,\beta)|=O(r^4)$, and $|A_s'|\leq m/r^3$ and $|B_s'|\leq n/r^3$ for each graph $A_s'\times B_s'$ in the collection. The number of pairs of points recorded in $\Pi(A,B,\alpha,\beta)$ is  $O(|\Pi(A,B,\alpha,\beta)|\cdot m/r^3\cdot n/r^3)$, which is $O(mn/r^2)$.
    This proves the lemma.
\end{proof}

The following theorem solves the complete BRS problem by running the algorithm of Lemma~\ref{lem:10} recursively until the problem size becomes $O(1)$. 

\begin{theorem}
    \label{theorem:PBRS}
     We can compute in $O(m^{2/3} n^{2/3} \cdot 2^{O(\log^* (m + n))} + m \log n + n \log m)$ time a collection $\Gamma(A, B, \alpha, \beta)=\{A_t \times B_t\ |\ A_t \subseteq A, B_t \subseteq B\}$ of edge-disjoint complete bipartite graphs that satisfy the conditions of Problem~\ref{prob:BRS} (with $\Pi(A, B, \alpha, \beta)=\emptyset$), with the following complexities: (1) $|\Gamma|=O(m^{2/3} n^{2/3} \cdot \log^* (m + n)+m+n)$; (2) $\sum_t |A_t|, \sum_t |B_t| = O(m^{2/3} n^{2/3} \cdot 2^{O(\log^* (m + n))} + m \log n + n \log m)$.
\end{theorem}
\begin{proof}
    To solve the complete BRS problem, the main idea is to apply the recurrence \eqref{formula:balanced} 
    recursively until the size of each subproblem becomes $O(1)$. 
    We first consider the symmetric case where $m = n$. By setting $r = n^{1/3} / \log n$ and applying \eqref{formula:balanced} with $m=n$, we obtain the following  
    \begin{gather}
        \label{formula:SymmetricBalanced}
        T(n, n) = O(n^{4/3}) + O(n^{4/3} / \log^4 n) \cdot T(\log^3 n, \log^3 n)
    \end{gather}
    Similarly, we have
     \begin{gather}
        \label{formula:SymmetricBalancedN}
        N(n, n) = O(n^{4/3}/\log^4 n) + O(n^{4/3} / \log^4 n) \cdot N(\log^3 n, \log^3 n)
    \end{gather}
    The recurrences solve to $T(n,n)=n^{4/3}\cdot 2^{O(\log^* n)}$ and $N(n,n)=O(n^{4/3}\cdot \log^* n)$. This means that in $n^{4/3}\cdot 2^{O(\log^* n)}$ time we can compute a collection $\Gamma(A,B,\alpha,\beta)=\{A_t\times B_t\ |\ A_t\subseteq A,B_t\subseteq B\}$ of $O(n^{4/3}\log^* n)$ edge-disjoint complete bipartite graphs, with $\sum_t |A_t|, \sum_t |B_t| = n^{4/3}\cdot 2^{O(\log^* n)}$, and it satisfies the conditions of Problem~\ref{prob:BRS} with $\Pi(A,B,\alpha,\beta)=\emptyset$. 

    We now consider the asymmetric case, i.e., $m \neq n$. We first assume $m\leq n$. Depending on whether $n < m^2$, there are two cases.
        \begin{enumerate}
            \item If $n < m^2$, we set $r = n / m$ so that $m / r = n / r^2$. We apply recurrence~\eqref{formula:PrimalRecurrence} and solve each subproblem of size $(m / r, n / r^2) = (m^2 / n, m^2 / n)$ by our above algorithm for the symmetric case, which results in $T(m, n) = O(n \log m + m^{2/3} n^{2/3} \cdot 2^{O(\log^* n)})$. Similarly, the number of 
            graphs in the produced collection is $O(m^{2/3} n^{2/3} \log^* n)$ and the total size of vertex sets of these graphs is $O(n \log m + m^{2/3} n^{2/3} \cdot 2^{O(\log^* n)})$.
            
            \item If $n \geq m^2$, then we simply apply recurrence~\eqref{formula:PrimalRecurrence} with $r=m$ and obtain $T(m,n)=O(m^2+n\log m)+O(m^2)\cdot T(1,n/m^2)$. Note that $T(1,n/m^2)$ can be solved in $O(n/m^2)$ time by brute force. Therefore, the recurrence solves to $T(m,n)=O(m^2+n\log m)$, which is $O(n\log m)$ as $n\geq m^2$. Similarly, the number of 
            of complete bipartite graphs in the generated collection is $O(n)$, and the total size of vertex sets of these graphs is $O(n \log m)$.
        \end{enumerate}
    In summary, if $m\leq n$, we can solve the complete BRS problem in $O(n \log m + m^{2/3} n^{2/3} \cdot 2^{O(\log^* n)})$ time, by generating $O(m^{2/3} n^{2/3} \log^* n+n)$ complete bipartite graphs whose vertex set size is bounded by $O(n \log m + m^{2/3} n^{2/3} \cdot 2^{O(\log^* n)})$. 
    
    If $m>n$, then the analysis is symmetric with the notation $m$ and $n$ flipped in the above complexities. The theorem is thus proved. 
\end{proof}

For comparison, Katz and Sharir~\cite{ref:KatzAn97} solved the complete BRS problem in $O((m^{2/3} n^{2/3} + m + n) \log m)$ time by producing $O(m^{2/3} n^{2/3} + m + n)$ complete bipartite graphs whose total vertex set size is $O((m^{2/3} n^{2/3} + m + n) \log m)$). Our result improves their runtime and vertex set size by almost a logarithmic factor with slightly more graphs produced. One may wonder whether Chan and Zheng's recent techniques~\cite{ref:ChanHo22} could be used to reduce the factor $2^{O(\log^*n)}$. It is not clear to us whether this is possible. Indeed, Chan and Zheng's techniques are mainly for solving point locations in line arrangements and in their problem they only need to locate a single cell of the arrangement that contains a point. In the point location step of our problem (i.e., computing the canonical sets $B(\sigma)$ in Lemma~\ref{lem:10}), however, we have to use hierarchical cutting and construct the canonical sets $B(\sigma)$ for each cell $\sigma$ that contains the point in every cutting $\Xi_i$, $1\leq i\leq k$ (i.e., our problem needs to place each point in $\Theta(\log r)$ cells and this placement operation already takes $\Theta(\log r)$ time).

\section{Distance selection}
\label{sec:DistanceSelection}

In this section, we present our algorithm for the distance selection problem. 
Let $P$ be a set of $n$ points in the plane. Define $\calE(P)$ as the set of distances of all pairs of points of $P$. Given an integer $1 \leq k \leq \binom{n}{2}$, the problem is to find the $k$-th smallest value in $\calE(P)$, denoted by $\delta^*$. 

Given any $\delta$, the {\em decision problem} is to determine whether  $\delta\geq \delta^*$. Wang~\cite{ref:WangUn23} recently gave an $O(n^{4/3})$ time algorithm that can compute the number of values of $\calE(P)$ at most $\delta$, denoted by $k_{\delta}$. Observe that $\delta\geq \delta^*$ if and only if $k_{\delta}\geq k$. Thus, using Wang's algorithm~\cite{ref:WangUn23}, the decision problem can be solved in $O(n^{4/3})$ time. 
We should point out that the $O(n^{4/3}\log^2 n)$ time algorithm of Katz and Sharir~\cite{ref:KatzAn97} for computing $\delta^*$ utilizes a decision algorithm of $O(n^{4/3}\log n)$ time. However, even if we  replace their decision algorithm by Wang's $O(n^{4/3})$ time algorithm, the runtime of the overall algorithm for computing $\delta^*$ is still $O(n^{4/3}\log^2 n)$ because other parts of the algorithm dominate the total time. To reduce the overall time to $O(n^{4/3}\log n)$, new techniques are needed, in addition to using the faster $O(n^{4/3})$ time decision algorithm. These new techniques include, for instance, Lemma~\ref{lem:10} for the partial BRS problem, as will be seen below. 

Before presenting the details of our algorithm, we first prove the following lemma, which is critical to our algorithm and is obtained by using Lemma~\ref{lem:10}. 

\begin{lemma}
\label{lem:20}
Given an interval $(\alpha,\beta]$, Problem~\ref{prob:BRS} with $A=P$ and $B=P$ can be solved in $O(n^{4/3})$ time by computing two collections $\Gamma(P,P,\alpha,\beta) = \{A_t \times B_t\ |\ A_t, B_t \subseteq P\}$ and $\Pi(P,P,\alpha,\beta) = \{A'_s \times B'_s\ |\ A'_s, B'_s \subseteq P\}$ with the following complexities: (1) $|\Gamma|=O(n^{4/3} / \log^4 \log n)$; (2) $\sum_t |A_t|, \sum_t |B_t| = O(n^{4/3})$; (3) $|\Pi|=O(n^{4/3} / \log^4 \log n)$; (4) $|A'_s|, |B'_s|=O(\log^3\log n)$, for each  $A_s'\times B'_s\in \Pi$.
\end{lemma}
\begin{proof}
    \label{proof:lem-DS-UTBRS}
    We first apply Lemma~\ref{lem:10} with $A = P$, $B = P$, and $r = n^{1/3} / \log n$. This constructs a collection $\Gamma_1 = \{A_t \times B_t\ |\ A_t, B_t \subseteq P\}$ of $O(n^{4/3} / \log^4 n)$ edge-disjoint complete bipartite graphs in $O(n^{4/3})$ time. The total size of vertex sets of these graphs is $O(n^{4/3})$, i.e., $\sum_t |A_t|, \sum_t |B_t| = O(n^{4/3})$. We also have a collection $\Pi_1 = \{A_s' \times B_s'\ |\ A_s', B_s' \subseteq P\}$ of $O(n^{4/3} / \log^4 n)$ edge-disjoint complete bipartite graphs that record uncertain point pairs,
    with $|A_s'|, |B_s'| = O(\log^3 n)$. 

    Hence, the number of uncertain pairs of points of $P$ (i.e., we do not know whether their distances are in $(\alpha, \beta]$) is $\sum_s|A'_s|\cdot |B'_s|=O(n^{4/3} \log^2 n)$. To further reduce this number, we apply Lemma~\ref{lem:10} on every pair $(A_s', B_s')$ of $\Pi_1$. More specifically, for each pair $(A_s', B_s')$ of $\Pi_1$, we apply Lemma~\ref{lem:10} with $A = A_s'$, $B = B_s'$, and $r = \log n / \log \log n$. This computes a collection $\Gamma_s$ of $O(\log^4 n / \log^4 \log n)$ edge-disjoint complete bipartite graphs in $O(\log^4 n)$ time; the total size of vertex sets of all graphs in $\Gamma_s$ is $O(\log^4 n)$. We also have a collection $\Pi_s$ of $O(\log^4 n / \log^4 \log n)$ edge-disjoint complete bipartite graphs.
    The size of each vertex set of each graph of $\Pi_s$ is bounded by $O(\log^3 \log n)$.
    The total time for Lemma~\ref{lem:10} on all pairs $(A_s',B_s')$ of $\Pi_1$ as above is $O(n^{4/3})$. We return $\Gamma_1 \cup \bigcup_s \Gamma_s$ as collection $\Gamma$, and $\bigcup_s \Pi_s$ as collection $\Pi$ in the lemma statement.
    As such, the complexities in the lemma statement hold. 
\end{proof}

In what follows, we describe our algorithm for computing $\delta^*$. 
Like Katz and Sharir's algorithm~\cite{ref:KatzAn97}, our algorithm proceeds in stages. 
Initially, we have $I_0 = (0, +\infty]$.
In each $j$-th stage, an interval $I_j = (\alpha_i, \beta_j]$ is computed from $I_{j - 1}$ such that $I_j$ must contain $\delta^*$ and the number of values of $\calE(P)$ in $I_j$ is a constant fraction of that in $I_{j-1}$. Specifically, 
we will prove that $|\calE(P)\cap I_j|=O(n^2 \rho^j)$ holds for each $j$, for some constant $\rho < 1$. Once $|\calE(P)\cap I_j|$ is no more than a threshold (to be given later; as will be seen later, this threshold is not constant, which is a main difference between our algorithm and Katz and Sharir's algorithm~\cite{ref:KatzAn97}), we will compute $\delta^*$ directly.
In the following we discuss the $j$-th stage of the algorithm. We assume that we have an interval $I_{j-1}=(\alpha_{j - 1}, \beta_{j - 1}]$ containing $\delta^*$. 

We first apply Lemma~\ref{lem:20} with $(\alpha, \beta] = (\alpha_{j - 1}, \beta_{j - 1}]$. This is another major difference between our algorithm and Katz and Sharir's algorithm~\cite{ref:KatzAn97}, where they solved the complete BRS problem, while we only solve a partial problem (this saves time by a logarithmic factor). Applying Lemma~\ref{lem:20} produces a collection $\Gamma_{j - 1} = \{A_t \times B_t\ |\ A_t, B_t \subseteq P\}$ of $O(n^{4/3} / \log^4 \log n)$ edge-disjoint complete bipartite graphs, with $\sum_t |A_t|, \sum_t |B_t| = O(n^{4/3})$, as well as another collection $\Pi_{j - 1}$ of $O(n^{4/3} / \log^4 \log n)$ graphs. By Lemma~\ref{lem:20} (3) and (4), the number of pairs of points of $P$ in $\Pi_{j-1}$ is $O(n^{4/3} \log^2 \log n)$.

If $\sum_t |A_t| \cdot |B_t| \leq n^{4/3} \log n$, which is our threshold, then this is the last stage of the algorithm and we compute $\delta^*$ directly by the following Lemma~\ref{lem:30}. Each edge of the graph in $\Gamma_{j-1}\cup \Pi_{j-1}$ connects two points of $P$; we say that the distance of the two points is {\em induced} by the edge. 

\begin{lemma}\label{lem:30}
If $\sum_t |A_t| \cdot |B_t| \leq n^{4/3} \log n$, then  $\delta^*$ can be computed in $O(n^{4/3}\log n)$ time. 
\end{lemma}
\begin{proof}
We first explicitly compute the set $S$ of  distances induced from edges of all graphs of $\Gamma_{j - 1}$ and $\Pi_{j - 1}$. Since $\sum_t |A_t| \cdot |B_t| \leq n^{4/3} \log n$ and the number of edges of all graphs of $\Pi_{j - 1}$ is $O(n^{4/3} \log^2 \log n)$, we have $|S| = O(n^{4/3} \log n)$ and $S$ can be computed in $O(n^{4/3} \log n)$ time by brute force. 

Then, we compute the number $k_{\alpha_{j-1}}$ of values of $\calE(P)$ that are at most $\alpha_{j-1}$, which can be done in $O(n^{4/3})$ time~\cite{ref:WangUn23}. Observe that $\delta^*$ is the $(k-k_{\alpha_{j-1}})$-th smallest value in $S$. Hence, using the linear time selection algorithm, we can find $\delta^*$ in $O(|S|)$ time, which is $O(n^{4/3} \log n)$.
\end{proof}

We now assume that $\sum_t |A_t| \cdot |B_t| > n^{4/3} \log n$. The rest of the algorithm for the $j$-th iteration takes $O(n^{4/3})$ time.
For each graph $A_t \times B_t \in \Gamma_{j - 1}$, if $|A_t|<|B_t|$, then we switch the name of $A_t$ and $B_t$, i.e., $A_t$ now refers to $B_t$ and $B_t$ refers to the original $A_t$. Note that this does not change the solution of the partial BRS produced by Lemma~\ref{lem:20} and it does not change the complexities of Lemma~\ref{lem:20} either. This name change is only for ease of the exposition. Now we have $|A_t|\geq |B_t|$ for each graph $A_t \times B_t \in \Gamma_{j - 1}$. Let $m_t = |A_t|$ and $n_t = |B_t|$. 

We partition each $A_t$ into $g=\lfloor m_t/n_t\rfloor$ subsets $A_{t1},A_{t2},\ldots, A_{tg}$ so that each subset contains $n_t$ elements except that the last subset $A_{tg}$ contains at least $n_t$ but at most $2n_t-1$ elements. 
Each pair $(A_{ti}, B_t)$, $1 \leq i \leq g$, can be viewed as a complete bipartite graph. As in \cite{ref:KatzAn97}, we construct a \emph{$d$-regular LPS-expander graph} $G_{ti}$ on the vertex set $A_{ti} \cup B_t$, for a constant $d$ to be fixed later.\footnote{A good summary of definitions and properties of expanders can be found in Section 2 of \cite{ref:KatzAn97}. Here it suffices for the reader to know the following property (which is needed in the proof of Lemma~\ref{lem:40}): If $X$ and $Y$ are two vertex subsets of a $d$-regular expander graph of $M$ vertices and there are fewer than $3M$ edges connecting points of $X$ and points of $Y$, then $|X| \cdot |Y| \leq 9 M^2 / d$.} The expander $G_{ti}$ has $O(|A_{ti}| + |B_t|)$ edges and can be computed in $O(|A_{ti}| + |B_t|)$ time~\cite{ref:KatzAn97, ref:LubotzkyEx86}. Let $G_t$ be the union of all these expander graphs  $G_{ti}$ over all $i=1,2,\ldots, g$. 
The construction of $G_t$ takes $\sum_{i = 1}^{g} O(|A_{ti}| + |B_t|) = O(|A_t| + \lfloor \frac{m_t}{n_t} \rfloor \cdot |B_t|) = O(|A_t|)$ time. Hence, computing all graphs $\{G_t\}_t$ for all $O(n^{4/3} / \log^4 \log n)$ pairs $A_t\times B_t$ in $\Gamma_{j - 1}$ takes $\sum_t O(|A_t|) = O(n^{4/3})$ time. The number of edges in $G_t$ is $O(|A_t| + |B_t|)$, and thus the number of edges in all graphs $\{G_t\}_t$ is $\sum_t O(|A_t| + |B_t|) = O(n^{4/3})$.

For each edge $(a, b)$ in graph $G_t$ that connects a point $a \in A_t$ and a point $b \in B_t$, we associate it with the interpoint distance $\lVert ab \rVert$. We compute all these distances for all graphs $\{G_t\}_t$ to form a set $S$. 
The size of $S$ is bounded by the number of edges in all graphs $\{G_t\}_t$, which is $O(n^{4/3})$. Note that all values of $S$ are in the interval $I_{j - 1}$.

One way we could proceed from here is to find the largest value $\delta_1$ of $S$ with $\delta_1<\delta^*$ and the smallest value $\delta_2$ with $\delta^*\leq \delta_2$, and then return $(\delta_1,\delta_2]$ as the interval $I_j$ and finish the $j$-th stage of the algorithm. Finding $\delta_1$ and $\delta_2$ could be done by binary search on $S$ using the linear time selection algorithm and the $O(n^{4/3})$ time decision algorithm. Then the runtime of this step would be $O(n^{4/3} \log n)$, resulting in a total of $O(n^{4/3}\log^2 n)$ time for the overall algorithm for computing $\delta^*$ since there are $O(\log n)$ stages. 
To improve the time, as in~\cite{ref:KatzAn97}, we use the ``Cole-like'' technique to reduce the number of calls to the decision algorithm to $O(1)$ in each stage, as follows. 



We assign a \emph{weight} to each value of $S$. Note that since each graph $G_{ti} \in G_t$ is a $d$-regular LPS-expander, the degree of $G_{ti}$ is $d$~\cite{ref:KatzAn97}. Hence, $G_{ti}$
has at most $(|A_{ti}| + |B_t|) \cdot d / 2$ edges and thus it contributes at most $(|A_{ti}| + |B_t|) \cdot d / 2$ values to $S$. We assign each distance induced from $G_{ti}$ a weight equal to $|A_{ti}| \cdot |B_t| / (|A_{ti}| + |B_t|)$. As such, the total weight of the values of $S$ is at most
    \begin{gather*}
        \sum_{t, i} (|A_{ti}| + |B_t|) \cdot \frac{d}{2} \cdot \frac{|A_{ti}| \cdot |B_t|}{|A_{ti}| + |B_t|} = \frac{d}{2}\cdot  \sum_{t, i} |A_{ti}| \cdot |B_t| =\frac{d}{2}\cdot   m_{j-1},
    \end{gather*}
    where $m_{j-1}=\sum_t |A_t| \cdot |B_t|$. 
    Recall that $m_{j-1} > n^{4/3}\log n$ and $|B_t| \leq |A_{ti}|$ in each $G_{ti}$. We can assume $n\geq 16$ so that $m_{j-1}\geq 16$. As such, we have the following bound for the weight of each value in $S$:
        $|A_{ti}| \cdot |B_t|/(|A_{ti}| + |B_t|) \leq |B_t| \leq \sqrt{|B_t| \cdot |A_{ti}|} \leq \sqrt{m_{j-1}} \leq m_{j-1}/4$.
        
We partition the values of $S$ into at most $2d$ intervals $\{I'_1, I'_2, ..., I'_{h}\}$, $1 \leq h \leq 2d$, such that the total weight of values in every interval is at least $m_{j-1} / 4$ and but at most $m_{j-1}/2$. The partition can be done in $O(|S|)$ time, which is $O(n^{4/3})$, using the linear time selection algorithm. Then, we invoke the decision algorithm $\log (2d) = O(1)$ times to find the interval $I'_l$ that contains $\delta^*$, for some $1\leq l\leq h$. We set $I_j = I'_l$. Since the decision algorithm is called $O(1)$ times, this step takes $O(n^{4/3})$ time. 
This finishes the $j$-th stage of the algorithm.

The following Lemma~\ref{lem:40} shows that the number of values of $\calE(P)$ in $I_j$ is a constant portion of that in $I_{j-1}$. This guarantees that the algorithm will finish in $O(\log n)$ stages since $|\calE(P)|=O(n^2)$. As each stage runs in $O(n^{4/3})$ time (except that the last stage takes $O(n^{4/3}\log n)$ time), the total time of the algorithm is $O(n^{4/3}\log n)$. 
\begin{lemma}\label{lem:40}
    \label{lemma:DS-correctness}
    There exists a constant $\rho$ with $0<\rho <1$ such that the number of values of $\calE(P)$ in $I_j$ is at most $\rho$ times the number of values of $\calE(P)$ in $I_{j - 1}$.
\end{lemma}
\begin{proof}
Define $n_j$ (resp., $n_{j-1}$) as the number of values of $\calE(P)$ in $I_j$ (resp., $I_{j-1}$). Our goal is to find a constant $\rho\in (0,1)$ so that $n_j\leq \rho\cdot n_{j-1}$ holds. 

Recall that $m_{j-1}$ is the number of distances induced from the graphs of $\Gamma_{j - 1}$. 
Define $m'_{j-1}$ as the number of distances induced from the graphs of $\Pi_{j - 1}$. 
Define $q_j$ (resp., $q'_j$) as the number of interpoint distances of $\calE(P)\cap I_j$ whose point pairs are recorded in $\Gamma_{j - 1}$ (resp., $\Pi_{j - 1}$). 
Note that all interpoint distances induced from graphs of $\Gamma_{j - 1}$ are in $I_{j-1}$. 
Hence, $m_{j-1}\leq n_{j-1}$. 
By definition, $n_j=q_j+q'_j$ and $q'_j \leq m'_{j-1}$. By Lemma~\ref{lem:20} (3) and (4), we have $m'_{j-1}= O(n^{4/3} \log^2 \log n)$. 

We make the following {\bf claim:} there exists a constant $\gamma \in (0, 1/3)$ such that $q_j \leq \gamma \cdot m_{j-1}$. Before proving the claim, we prove the lemma using the claim. 

As this is not the last stage of the algorithm (since otherwise $\delta^*$ would have already been computed without producing interval $I_j$), it holds that $m_{j-1} > n^{4/3} \log n$.
Since $m'_{j-1} = O(n^{4/3} \log^2 \log n)$, there exists a constant $c' \in (0, 1/3)$ such that $\frac{m'_{j-1}}{m_{j-1}}\leq c'$ when $n$ is sufficiently large. As $n_j=q_j+q'_j$,  $q'_j \leq m'_{j-1}$, and $m_{j-1}\leq n_{j-1}$, we can obtain the following using the above claim:
    \begin{gather*}
        n_j = q_j + q'_j \leq q_j + m'_{j-1}\leq \gamma \cdot m_{j-1} + c' \cdot m_{j-1} \leq (\gamma + c') \cdot m_{j-1}\leq (\gamma + c') \cdot n_{j-1}.
    \end{gather*}

Set $\rho = \gamma + c'$. Since both $\gamma$ and $c'$ are in $(0,1/3)$, we have $\rho \in (0, 2/3)$ and $n_j\leq \rho\cdot n_{j-1}$. This proves the lemma.

\paragraph{Proof of the claim.}
We now prove that there exists a constant $\gamma \in (0, 1/3)$ such that $q_j \leq \gamma \cdot m_{j-1}$. The proof is similar to that in \cite{ref:KatzAn97}.

Consider a pair $(A_{ti}, B_t)$, $1 \leq i \leq g$, obtained in our algorithm. Some edges of the graph $G_{ti}$ induce interpoint distances in $S$, which may be in $I_j$. We partition all such graphs $G_{ti}$ into two sets. 
Let $\calG_1$ denote the set of those graphs $G_{ti}$ that contribute fewer than $3(|A_{ti}| + |B_t|)$ interpoint distances in $S\cap I_j$, and $\calG_2$ the set of the rest of such graphs (each of them contributes at least $3(|A_{ti}| + |B_t|)$ interpoint distances in $S\cap I_j$).

\paragraph{The set $\calG_1$.}
We first consider set $\calG_1$. For a graph $G_{ti} \in \calG_1$ built on pair $(A_{ti}, B_t)$, let $\calD_{ti}$ be the set of annuli centered at points of $A_{ti}$ with radii $\alpha_j$ and $\beta_j$ (recall that $I_j=(\alpha_j,\beta_j]$). For the purpose of analysis only, we construct a $1/r$-cutting $\Xi$ for the boundary circles of the annuli in $\calD_{ti}$, where $r$ is a constant to be specified later. This partitions the plane into $O(r^2)$ cells such that each cell intersects at most $O(|\calD_{ti}| / r)$ boundary circles of annuli in $\calD_{ti}$.

For each cell $\sigma \in \Xi$, let $B_t(\sigma)$ denote the set of points of $B_t$ inside $\sigma$, $\calD_{ti}(\sigma)$ the set of annuli of $\calD_{ti}$ that fully contains $\sigma$, and $\calD_{ti}'(\sigma)$ the set of annuli of $\calD_{ti}$ that have at least one boundary circle intersecting $\sigma$. Let $N_{ti}$ denote the number of interpoint distances between points of $A_{ti}$ and points of $B_t$ that are in $I_j$. Then we have
    \begin{gather*}
        N_{ti} \leq \sum_{\sigma \in \Xi} |\calD_{ti}(\sigma)| \cdot |B_t(\sigma)| + \sum_{\sigma \in \Xi} |\calD_{ti}'(\sigma)| \cdot |B_t(\sigma)|
    \end{gather*}
    Since the number of annuli of $\calD_{ti}$ that intersect a cell $\sigma \in \Xi$ is $O(|\calD_{ti}| / r)$ and $|\calD_{ti}| = |A_{ti}|$, we have $|\calD_{ti}'(\sigma)| = O(|A_{ti}| / r)$. Using $\sum_{\sigma \in \Xi} |B_t(\sigma)| = |B_t|$, we can derive
    \begin{gather*}
        \sum_{\sigma \in \Xi} |\calD_{ti}'(\sigma)| \cdot |B_t(\sigma)| = O\left( \frac{|A_{ti}| \cdot |B_t|}{r} \right)
    \end{gather*}

    Now we consider $\sum_{\sigma \in \Xi} |\calD_{ti}(\sigma)| \cdot |B_t(\sigma)|$. Let $A_{ti}(\sigma) \subseteq A_{ti}$ denote the set of centers of the annuli of $\calD_{ti}(\sigma)$. 
    For any point $a\in A_{ti}(\sigma)$ and $b\in B_t(\sigma)$, 
    their distance $\lVert ab \rVert$ is in $I_{j}$ by the definition of $\calD_{ti}(\sigma)$. If an edge connecting $a$ and $b$ exists in graph $G_{ti}$, then $\lVert ab \rVert$ must be in $S$ and thus is in $I_j$ as well, i.e., such an edge of $G_{ti}$ contributes a value in $S\cap I_j$. Since $G_{ti}$ is in $\calG_1$, it has fewer than $3(|A_{ti}| + |B_t|)$ edges whose induced interpoint distances are in $I_j$, which implies that the number of edges of $G_{ti}$ connecting points of $A_{ti}(\sigma)$ and points of $B_t(\sigma)$ in $G_{ti}$ is smaller than $3(|A_{ti}| + |B_t|)$. According to Corollary 2.5 in~\cite{ref:KatzAn97}, if $X$ and $Y$ are two vertex subsets of a $d$-regular expander graph of $M$ vertices and there are fewer than $3M$ edges connecting points of $X$ and points of $Y$, then $|X| \cdot |Y| \leq 9 M^2 / d$. Applying this result (with $X = A_{ti}(\sigma)$, $Y = B_t(\sigma)$, and $M = |A_{ti}| + |B_t|$), we can derive the following 
    \begin{gather*}
        \sum_{\sigma \in \Xi} |\calD_{ti}(\sigma)| \cdot |B_t(\sigma)| \leq O(r^2) \cdot \frac{9 (|A_{ti}| + |B_t|)^2}{d} = O\left( \frac{r^2 (|A_{ti}| + |B_t|)^2}{d} \right)
    \end{gather*}

    In summary, we have,
    \begin{gather*}
    N_{ti} = O\left( \frac{|A_{ti}| \cdot |B_t|}{r} \right) + O\left( \frac{r^2 (|A_{ti}| + |B_t|)^2}{d} \right).
    \end{gather*}
    Since $|B_t| \leq |A_{ti}| \leq 2|B_t|$ by our partition of set $A_t$, we have $(|A_{ti}| + |B_t|)^2 \leq 5 |A_{ti}| \cdot |B_t|$, which leads to
    \begin{gather*}
        N_{ti} = O \left( \left[\frac{1}{r} + \frac{r^2}{d}\right] \cdot |A_{ti}| \cdot |B_t| \right)
    \end{gather*}

    By setting $r = d^{1/3}$ and $c$ to be appropriately proportional to $1/d^{1/3}$, we obtain $N_{ti} \leq c \cdot |A_{ti}| \cdot |B_t|$. Summing up all these inequalities for all graphs $G_{ti}$ in set $\calG_1$ leads to $N(\calG_1) \leq c \cdot \sum_{G_{ti}\in \calG_1} |A_{ti}| \cdot |B_t|$, where $N(\calG_1)$ is the number of distances between points of $A_{ti}$ and points of $B_t$ that are in $I_j$ for all graphs $G_{ti} \in \calG_1$. Since $\sum_{G_{ti}\in \calG_1}|A_{ti}| \cdot |B_t|\leq \sum_t|A_t|\cdot |B_t|=m_{j-1}$, we obtain $N(\calG_1)\leq c\cdot m_{j-1}$. 

\paragraph{The set $\calG_2$.}
    We now consider the set $\calG_2$. Since each graph $G_{ti} \in \calG_2$ contributes at least $3(|A_{ti}| + |B_t|)$ interpoint distances in $S\cap I_j$, $G_{ti}$ contributes at least $3|A_{ti}| \cdot |B_t|$ to the total weight of distances in $S\cap I_j$. Recall that the total weight of distances in $S\cap I_j$ is at most $m_{j-1}/2$ by our algorithm, thus we have $\sum_{G_{ti} \in \calG_2} |A_{ti}| \cdot |B_t| \leq m_{j-1} / 6$. Let $N(\calG_2)$ denote the number of distances between points of $A_{ti}$ and points of $B_t$ that are in $I_j$ for all graphs $G_{ti} \in \calG_2$. We have $N(\calG_2) \leq \sum_{G_{ti} \in \calG_2} |A_{ti}| \cdot |B_t|$ since $I_{j} \subseteq I_{j - 1}$. Therefore, $N(\calG_2)\leq m_{j-1}/6$.

\paragraph{Summary.}
    By definition, $q_j = N(\calG_1) + N(\calG_2)$. As $N(\calG_1)\leq c\cdot m_{j-1}$ and $N(\calG_2)\leq m_{j-1}/6$, we can derive 
    \begin{align*}
        q_j  = N(\calG_1) + N(\calG_2) \leq c \cdot m_{j-1}+ \frac{1}{6} \cdot m_{j-1}
        = (c + \frac{1}{6}) \cdot m_{j-1}.
    \end{align*}
    Let $\gamma = c + 1/6$. Then $\gamma < 1/3$ if $d$ is sufficiently large. As such, we have $q_j \leq \gamma \cdot m_{j-1}$ for a constant $\gamma \in (0, 1/3)$. The claim is thus proved. 
\end{proof}

We conclude with the following result.
\begin{theorem}
    \label{theorem:DS}
    Given a set $P$ of $n$ points in the plane and an integer $1 \leq k \leq \binom{n}{2}$, the $k$-th smallest interpoint distance of $P$ can be computed in $O(n^{4/3} \log n)$ time.
\end{theorem}

Note that once $\delta^*$ is computed, one can find a pair of points of $P$ whose distance is equal to $\delta^*$ in additional $O(n^{4/3})$ time~\cite{ref:WangUn23}. 

\paragraph{A bipartite version.} Our algorithm can be easily extended to the following {\em bipartite version} of the distance selection problem: Given a set $A$ of $m$ points and a set $B$ of $n$ points in the plane, and an integer $1 \leq k \leq mn$, compute the $k$-th smallest interpoint distance $\delta^*$ in the set $\{\lVert ab\rVert\ |\ a\in A, b\in B\}$. The decision problem can be solved in $O(m^{2/3} n^{2/3} + m \log n + n \log m)$ time~\cite{ref:WangUn23}.
To adapt our algorithm to compute $\delta^*$, each stage of the algorithm still computes an interval $I_j$ as before. In the $j$-th stage, we solve the partial BRS problem for $A$ and $B$ with respect to the interval $I_{j - 1}$. We can obtain a result similar to Lemma~\ref{lem:20} (by using Lemma~\ref{lem:20} as a subroutine in an analogous way to Theorem~\ref{theorem:PBRS} for dealing with the asymmetric case). More specifically, if $m \leq n < m^2$ (resp. $n \leq m < n^2$), we construct a hierarchical cutting and process those unsolved subproblems by applying Lemma~\ref{lem:20} with $r = n/m$ (resp. $r = m/n$). If $n \geq m^2$ or $m \geq n^2$, we construct a hierarchical cutting and process those unsolved subproblems in a straightforward manner. As such, we can obtain a collection $\Gamma$ of $O(m^{2/3} n^{2/3} / \log^4 \log(m^2 / n) + m^{2/3} n^{2/3} / \log^4 \log(n^2 / m) + m + n)$ edge-disjoint complete bipartite graphs that record some pairs of $A \times B$ whose interpoint distances are in $I_{j - 1}$. The total size of vertex sets of all graphs in $\Gamma$ is $O(m^{2/3} n^{2/3} + m \log n + n \log m)$. We also have another collection $\Pi$ of edge-disjoint complete bipartite graphs that record a total of
$O(m^{2/3} n^{2/3} \log^2 \log (m + n))$ uncertain pairs of $A \times B$, i.e., we do not know whether their distances are in $I_{j - 1}$. The total runtime is $O(m^{2/3} n^{2/3} + m \log n + n \log m)$. We compute the number of interpoint distances induced from collection $\Gamma$. If this number is at most $(m^{2/3} n^{2/3} + m \log n + n \log m) \log (m + n)$, then this is the last stage of the algorithm and we compute $\delta^*$ directly. Otherwise, we use the ``Cole-like'' technique to perform a binary search on the interpoint distances induced from the expander graphs that are built on vertex sets of the graphs in $\Gamma$, which calls the decision algorithm $O(1)$ times. 
The algorithm will finish within $O(\log (m + n))$ stages by similar analysis to Lemma~\ref{lemma:DS-correctness}. As such, the bipartite distance selection problem can be solved in $O((m^{2/3} n^{2/3} + m \log n + n \log m) \log (m + n))$ time. 


\section{Two-sided discrete Fr\'{e}chet distance with shortcuts}
\label{sec:TS-DFD}

In this section, we show that our techniques in Section~\ref{sec:DistanceSelection} can be used to solve the two-sided DFD problem. Let $A = \{a_1, a_2, ..., a_m\}$ and $B = \{b_1, b_2, ..., b_n\}$ be two sequences of points in the plane. Consider two frogs connected by an inelastic leash, initially placed at $a_1$ and $b_1$, respectively. Each frog is allowed to jump forward at most one step in one move, i.e., if the first frog is currently at $a_i$, then in the next move it can either jump to $a_{i+1}$ or stay at $a_i$. Note that frogs are not allowed to go backwards. The \emph{discrete Fr\'{e}chet distance} (or DFD for short) is defined as the minimum length of the inelastic leash that allows two frogs to reach their destinations, i.e., $a_m$ and $b_n$, respectively.

Because the Fr\'{e}chet distance is very sensitive to outliers, to reduce the sensitivity, DFD with outliers have been proposed~\cite{ref:AvrahamTh15}. Specifically, if we allow the $A$-frog to jump from its current point to any of its succeeding points in each move but $B$-frog has to traverse all points in $B$ in order plus one restriction that only one frog is allowed to jump in each move (i.e., in each move one of the frogs must stay still), then this problem is called \emph{one-sided discrete Fr\'{e}chet distance with shortcuts} (or {\em one-sided DFD} for short), where the goal is to compute the minimum length of the inelastic leash that allows two frogs to reach their destinations.
If we allow both frogs to skip points in their sequences (but again with the restriction that only one frog is allowed to jump in each move), then problem is called {\em two-sided DFD}. 

We focus on the two-sided DFD in this section while the one-sided version will be treated in the next section. Let $\delta^*$ denote the optimal objective value, i.e., the minimum length of the leash. Avraham, Filtser, Kaplan, Katz, and Sharir~\cite{ref:AvrahamTh15} presented an algorithm that can compute $\delta^*$ in $O((m^{2/3} n^{2/3} + m + n) \log^3 (m + n))$ time. In what follows, we show that our techniques in Section~\ref{sec:DistanceSelection} can improve their algorithm to $O((m^{2/3} n^{2/3} \cdot 2^{O(\log^* (m + n))} + m \log n + n \log m) \log (m + n))$ time, roughly a factor of $O(\log^2(m+n))$ faster.

To solve the problem, the authors of \cite{ref:AvrahamTh15} first proposed an algorithm to solve the decision problem, i.e., given any $\delta$, decide whether $\delta^*\leq \delta$; the algorithm runs in $O((m^{2/3} n^{2/3} + m + n) \log^2 (m + n))$ time. Then, to compute $\delta^*$, the authors of \cite{ref:AvrahamTh15} used the bipartite version of the distance selection algorithm from Katz and Sharir~\cite{ref:KatzAn97} for point sets $A$ and $B$ together with their decision algorithm to do binary search on the interpoint distances between points in $A$ and those in $B$, i.e., in each iteration, using the distance selection algorithm to find the $k$-th smallest distance $\delta_k$ for an appropriate $k$ and then call the decision algorithm on $\delta_k$ to decide which way to search. As both the distance selection algorithm~\cite{ref:KatzAn97} and the decision algorithm run in $O((m^{2/3} n^{2/3} + m + n) \log^2 (m + n))$ time, computing $\delta^*$ takes $O((m^{2/3} n^{2/3} + m + n) \log^3 (m + n))$ time. 

In what follows, we first show that the runtime of their decision algorithm can be reduced by a factor of roughly $O(\log^2(m+n))$ using our result in Theorem~\ref{theorem:PBRS} for the complete BRS problem, and then discuss how to improve the optimization algorithm for computing $\delta^*$. 

\paragraph{Improving the decision algorithm.}
The basic idea of the decision algorithm in \cite{ref:AvrahamTh15} is to consider a matrix $M$ whose rows and columns correspond to points in sequences $A$ and $B$, respectively. Each entry $M(i, j)$ of $M$ is $1$ if $\lVert a_i b_j \rVert \leq \delta$, and $0$ otherwise. One can determine whether there exists a path from $M(1, 1)$ to $M(m, n)$ in $M$ that only consists of value $1$ by performing ``upward'' and ``rightward'' moves. The matrix $M$ is not computed explicitly. 
The algorithm first performs a complete BRS with $\alpha = 0$ and $\beta = \delta$ using a result from \cite{ref:KatzAn97} on $A$ and $B$, which generates a collection $\Gamma = \{A_t \times B_t\}_t$ of complete bipartite graphs that record all pairs of $A \times B$ whose interpoint distances are at most $\delta$ in $O((m^{2/3} n^{2/3} + m + n) \log (m + n))$ time, with $\sum_t|A_t|, \sum_t|B_t|=O((m^{2/3} n^{2/3} + m + n) \log (m + n))$. 
Each edge of these graphs corresponds to an entry of value $1$ in $M$. Then for each graph $A_t \times B_t\in \Gamma$, points of $A_t$ and $B_t$ are sorted by their index order into lists $\calL_{A_t}$ and $\calL_{B_t}$, respectively. The sorting takes $O((m^{2/3} n^{2/3} + m + n) \log^2 (m + n))$ time in total. 
With these information in hand, the rest of the algorithm runs in time linear in the total size of vertex sets of graphs in $\Gamma$, which is $O((m^{2/3} n^{2/3} + m + n) \log (m + n))$. 

We can improve their decision algorithm by applying our complete BRS result in Theorem~\ref{theorem:PBRS}. Specifically, applying Theorem~\ref{theorem:PBRS} will produce in $O(m^{2/3} n^{2/3} \cdot 2^{O(\log^* (m + n))} + m \log n + n \log m)$ time a collection $\Gamma$ of complete bipartite graphs that record all pairs of $A \times B$ whose interpoint distances are at most $\delta$. 
To reduce the time on the sorting step, 
when computing the canonical subsets $B(\sigma)$ in Lemma~\ref{lem:10}, we process points of $B$ following their index order. Similarly, when computing the canonical sets of $A_{\sigma}$, we process the circles of $\calC_{\sigma'}$ following the index order of their centers in $A$. This ensures that points in each $A_t$ and each $B_t$ are sorted automatically during the construction, i.e., lists $\calL_{A_t}$ and $\calL_{B_t}$ are available once the algorithm of Theorem~\ref{theorem:PBRS} is done. The rest of the algorithm follows exactly the same as the algorithm in~\cite{ref:AvrahamTh15}, which takes time proportional to the total size of vertex sets of graphs in $\Gamma$, i.e., $O(m^{2/3} n^{2/3} \cdot 2^{O(\log^* (m + n))} + m \log n + n \log m)$ by Theorem~\ref{theorem:PBRS}. As such, the total time of the new decision algorithm is $O(m^{2/3} n^{2/3} \cdot 2^{O(\log^* (m + n))} + m \log n + n \log m)$. 

\paragraph{Improving the optimization algorithm.}
With our new $O((m^{2/3} n^{2/3} + m \log n + n \log m) \log (m + n))$ time bipartite distance selection algorithm in Section~\ref{sec:DistanceSelection} and the above faster decision algorithm, following the same binary search scheme as discussed above, $\delta^*$ can be computed in $O((m^{2/3} n^{2/3} + m \log n + n \log m) \log^2 (m + n))$ time, a logarithmic factor improvement over the result of~\cite{ref:AvrahamTh15}. Notice that the time is dominated by the calls to the bipartite distance selection algorithm.

To further improve the algorithm, an observation is that we do not have to call the distance selection algorithm as an oracle and instead we can use that algorithm as a framework and replace the decision algorithm of the distance selection problem by the decision algorithm of the two-sided DFD problem. This will roughly reduce another logarithmic factor. 
The proof of the following theorem provides the details about this idea.

\begin{theorem}
\label{theorem:TS-DFD}
Given two sequences of points $A = (a_1, a_2, ..., a_m)$ and $B = (b_1, b_2, ..., b_n)$ in the plane, the two-sided DFD problem can be solved in $O((m^{2/3} n^{2/3} \cdot 2^{O(\log^* (m + n))} + m \log n + n \log m) \log (m + n))$ time.
\end{theorem}
\begin{proof}
    \label{proof:theorem-TS-DFD}
    Following our distance selection algorithm, we run in stages and each $j$-th stage will compute an interval $I_j$ that contains $\delta^*$.
    In the $j$-th stage, we first perform the partial BRS on point sets $A$ and $B$ with respect to interval $I_{j - 1}$, in the same way as before. This produces a collection $\Gamma$ of $(m^{2/3} n^{2/3} / \log^4 \log(m^2 / n) + m^{2/3} n^{2/3} / \log^4 \log(n^2 / m) + m + n)$ edge-disjoint complete bipartite graphs that record some pairs of $A \times B$ whose interpoint distances are in $I_{j - 1}$. The total size of vertex sets of all graphs in $\Gamma$ is $O(m^{2/3} n^{2/3} + m \log n + n \log m)$. In addition, we also have a collection $\Pi$ of complete bipartite graphs representing $O(m^{2/3} n^{2/3} \log^2 \log (m + n))$ uncertain pairs of $A \times B$.
    The total runtime is $O(m^{2/3} n^{2/3} + m \log n + n \log m)$.

    We next compute the number $n_{\Gamma}$ of distances induced from the graphs of $\Gamma$. If $n_{\Gamma}$ is larger than the threshold $\tau = (m^{2/3} n^{2/3} + m \log n + n \log m) \log (m + n)$, then we use the ``Cole-like'' technique to perform a binary search on the interpoint distances induced from the expander graphs that are built on the vertex sets of the graphs in $\Gamma$, which calls the decision algorithm $O(1)$ times. 
    The runtime for this stage is $O(m^{2/3} n^{2/3} \cdot 2^{O(\log^* (m + n))} + m \log n + n \log m)$. If $n_{\Gamma}\leq \tau$, then we reach the last stage of the algorithm and we can compute $\delta^*$ as follows. We compute the interpoint distances induced from the graphs in $\Gamma$ and $\Pi$. The total number of such distances is $O((m^{2/3} n^{2/3} + m \log n + n \log m) \log (m + n))$. Using the decision algorithm and the linear time selection algorithm, a binary search on these interpoint distances is performed to compute $\delta^*$, which takes $O((m^{2/3} n^{2/3} \cdot 2^{O(\log^* (m + n))} + m \log n + n \log m) \log (m + n))$ time as the decision algorithm is called $O(\log (m+n))$ times. The algorithm finishes within $O(\log (m + n))$ stages by an analysis similar to Lemma~\ref{lemma:DS-correctness} (indeed, the proof of Lemma~\ref{lemma:DS-correctness} does not rely on which decision algorithm is used).

    In summary, the total runtime for computing $\delta^*$ is bounded by $O((m^{2/3} n^{2/3} \cdot 2^{O(\log^* (m + n))} + m \log n + n \log m) \log (m + n))$.
\end{proof}

\paragraph{A general (deterministic) algorithmic framework.} The algorithm of Theorem~\ref{theorem:TS-DFD} can be made into a general algorithmic framework for solving geometric optimization problems involving interpoint distances in the plane. Specifically, suppose we have an optimization problem $\calP$ whose optimal objective value $\delta^*$ is equal to $\lVert ab\rVert$ for a point $a\in A$ and a point $b\in B$, with $A$ as a set of $m$ points and $B$ as a set of $n$ points in the plane. The goal is to compute $\delta^*$. Suppose that we have a decision algorithm that can determine whether $\delta\geq\delta^*$ in $T_D$ time for any $\delta$. Then, we can compute $\delta^*$ by applying exactly the same algorithm of Theorem~\ref{theorem:TS-DFD} except that we use the decision algorithm for $\calP$ instead. The total time of the algorithm is $O((m^{2/3} n^{2/3} + m \log n + n \log m + T_D)\cdot \log (m + n))$. Note that in the case  $T_D=o((m^{2/3} n^{2/3} + m \log n + n \log m)\log(m+n))$ this is faster than the traditional binary search approach by repeatedly invoking the distance selection algorithm.

\begin{theorem}
    \label{theorem:framework1}
    Given two sets $A$ and $B$ of $m$ and $n$ points respectively in the plane, any geometric optimization problem whose optimal objective value is equal to the distance between a point of $a\in A$ and a point of $b\in B$ can be solved in $O((m^{2/3} n^{2/3} + m \log n + n \log m + T_D)\cdot \log (m + n))$ time, where $T_D$ is the time for solving the decision version of the problem.
\end{theorem}


\section{One-sided discrete Fr\'{e}chet distance with shortcuts}
\label{sec:OS-DFD}

In this section, we consider the one-sided DFD problem, defined in Section~\ref{sec:TS-DFD}. Let $\delta^*$ denote the optimal objective value. Avraham, Filtser, Kaplan, Katz, and Sharir~\cite{ref:AvrahamTh15} proposed an a randomized algorithm of $O((m + n)^{6/5 + \epsilon})$ expected time.
We show that using our result in Lemma~\ref{lem:10} for the partial BRS problem the runtime of their algorithm can be reduced to $O((m + n)^{6/5} \log^{7/5} (m + n))$. 

Define $\calE(A,B)=\{\lVert ab\rVert \ | \ a\in A, b\in B\}$. It is known that $\delta^*\in \calE(A,B)$~\cite{ref:AvrahamTh15}. 
The decision problem is to decide whether $\delta\geq \delta^*$ for any $\delta$. 
The authors~\cite{ref:AvrahamTh15} first solved the decision problem in 
$O(m + n)$ (deterministic) time. To compute $\delta^*$, their algorithm has two main procedures. 

The first main procedure computes an interval $(\alpha, \beta]$ that is guaranteed to contain $\delta^*$, and in addition, with high probability the interval contains
at most $L$ values of $\calE(A,B)$, given any $1\leq L\leq mn$; the algorithm runs
in $O((m + n)^{4/3 + \epsilon}/L^{1/3} +(m + n)\log(m + n)\log\log (m+n))$ time, for any $\epsilon>0$. 
More specifically, during the course of the algorithm, an interval $(\alpha, \beta]$ containing $\delta^*$ is maintained; initially $\alpha=0$ and $\beta=\infty$. 
In each iteration, the algorithm first determines, through random sampling, whether the number of values of $\calE(A,B)$ in $(\alpha, \beta]$ is at most $L$ with high probability. If so, the algorithm stops by returning the current interval $(\alpha,\beta]$. Otherwise, a subset $R$ of $O(\log (m+n))$ values of $\calE(A,B)$ is sampled which contains with high probability an approximate median (in the middle three quarters) among the values of $\calE(A,B)$ in $(\alpha, \beta]$. A binary search guided by the decision algorithm is performed to narrow down the interval $(\alpha, \beta]$; the algorithm then proceeds with the next iteration. As such, after $O(\log (m+n))$ iterations, the algorithm eventually returns an interval $(\alpha, \beta]$ with the property discussed above.

The second main procedure is to find $\delta^*$ from $\calE(A,B)\cap (\alpha,\beta]$. This is done by using a \emph{bifurcation tree} technique (Lemma 4.4~\cite{ref:AvrahamTh15}), whose runtime relies on $L'$, the true number of values of $\calE(A,B)$ in $(\alpha,\beta]$. As it is possible that $L'>L$, if the algorithm detects that case happens, then the first main procedure will run one more round from scratch. As $L'<L$ holds with high probability, the expected number of rounds is $O(1)$. If $L'\leq L$, the runtime of the second main procedure is bounded by $O((m + n) L^{1/2} \log(m + n))$. 

As such, the expected time of the algorithm is $O((m + n)^{4/3 + \epsilon}/L^{1/3} +(m + n)\log(m + n)\log\log (m+n) + (m + n) L^{1/2} \log(m + n))$. Setting $L$ to $O((m+n)^{2/5+\epsilon})$ for another small $\epsilon>0$, the time can be bounded by $O((m + n)^{6/5 + \epsilon})$. 

\paragraph{Our improvement.}
We can improve the runtime of the first main procedure by a factor of $O((m + n)^{\epsilon})$, which leads to the improvement of overall algorithm by a similar factor. To this end, by applying Lemma~\ref{lem:10} with $r=(\frac{m + n}{L})^{1/3}$, we first have the following corollary, which improves Lemma 4.1 in~\cite{ref:AvrahamTh15} (which is needed in the first main procedure).

\begin{corollary}
\label{lemma:OS-DFD-UTBRS}
Given a set $A$ of $m$ points and a set $B$ of $n$ points in the plane, an interval $(\alpha, \beta]$, and a parameter $1 \leq L \leq mn$, we can compute in $O((m + n)^{4/3} / L^{1/3} \cdot \log (\frac{m + n}{L}))$ time two collections $\Gamma(A, B, \alpha, \beta)= \{A_t \times B_t\ |\ A_t \subseteq A, B_t \subseteq B\}$ 
     and $\Pi(A, B, \alpha, \beta)= \{A'_s \times B'_s\ |\ A'_s \subseteq A, B'_s \subseteq B\}$ of edge-disjoint complete bipartite graphs that satisfy the conditions of Problem~\ref{prob:BRS},
     with the following complexities: (1) $|\Gamma|=O((\frac{m + n}{L})^{4/3})$; (2) $\sum_t |A_t|, \sum_t |B_t| = O((m + n)^{4/3} / L^{1/3} \cdot \log (\frac{m + n}{L}))$; (3) $|\Pi|=O((\frac{m + n}{L})^{4/3})$; (4) $|A_s'| = O(\frac{mL}{m + n})$ and $|B_s'| = O(\frac{nL}{m + n})$ for each $A_s' \times B_s' \in \Pi$; (5) the number of pairs of points recorded in $\Pi$ is $O((m + n)^{4/3} L^{2/3})$.
\end{corollary}

Replacing Lemma 4.1 in~\cite{ref:AvrahamTh15} by our results in Corollary~\ref{lemma:OS-DFD-UTBRS} and following the rest of the algorithm in~\cite{ref:AvrahamTh15} leads to an algorithm to compute $\delta^*$ in $O((m + n)^{6/5} \log^2 (m + n))$ expected time. To make the paper more self-contained, we present some details below. Also, we put the discussion in the context of a more general algorithmic framework (indeed, a recent result of Katz and Sharir~\cite{ref:KatzEf21arXiv} already gave such a framework; here we improve their result by a factor of $O((m+n)^{\epsilon})$ due to Corollary~\ref{lemma:OS-DFD-UTBRS}). 

\paragraph{A general (randomized) algorithmic framework.}
Suppose we have an optimization problem $\calP$ whose optimal objective value $\delta^*$ is equal to $\lVert ab\rVert$ for a point $a\in A$ and a point $b\in B$, with $A$ as a set of $m$ points and $B$ as a set of $n$ points in the plane. The goal is to compute $\delta^*$. Suppose that we have a decision algorithm that can determine whether $\delta\geq\delta^*$ in $T_D$ time for any $\delta$. With the result from Corollary~\ref{lemma:OS-DFD-UTBRS}, we have the following lemma. Define $\calE(A,B)$ in the same way as above. 

\begin{lemma}
\label{theorem:LDistancesFramework}
Given any $1\leq L\leq mn$, there is a randomized algorithm that can compute an interval $(\alpha, \beta]$ that contains $\delta^*$ and with high probability contains at most $L$ values of $\calE(A,B)$; the expected time of the algorithm is $O((m + n)^{4/3} / L^{1/3} \cdot \log^2 (m + n) + T_D \cdot \log (m + n) \cdot \log\log (m + n))$.
\end{lemma}
\begin{proof}
\label{proof:theorem-LDistancesFramework}
We maintain an interval $(\alpha, \beta]$ (which is initialized to $(0, +\infty]$) containing $\delta^*$ and shrink it iteratively. In each iteration, we first invoke Corollary~\ref{lemma:OS-DFD-UTBRS} to obtain two collections $\Gamma(A, B, \alpha, \beta)$ and $\Pi(A, B, \alpha, \beta)$ of complete bipartite graphs in $O((m + n)^{4/3} / L^{1/3} \cdot \log (\frac{m + n}{L}))$ time. In particular, the graphs of $\Pi(A, B, \alpha, \beta)$ record uncertain point pairs of $A \times B$ that we do not know whether their distances are in $(\alpha, \beta]$. The total number of these uncertain pairs is $M = O((m + n)^{4/3} L^{2/3})$.

Let $S_1$ (resp., $S_2$) denote the set of interpoint distances recorded in collection $\Gamma(A, B, \alpha, \beta)$ (resp., $\Pi(A, B, \alpha, \beta)$). Note that $|S_2|=M$ and all values of $S_1$ are in $(\alpha, \beta]$ while some values of $S_2$ may not be in $(\alpha, \beta]$. Define $S'_2$ to be the subset of distances of $S_2$ that lie in $(\alpha, \beta]$. We need to determine the number of distances of $S_1 \cup S_2$ that lie in $(\alpha, \beta]$, i.e., determine $|S_1|+|S_2'|$. To this end, as $|S_1|=\sum_t|A_t|\cdot |B_t|$ and $\sum_t |A_t|, \sum_t |B_t| = O((m + n)^{4/3} / L^{1/3} \cdot \log (\frac{m + n}{L}))$, $|S_1|$ can be easily computed in $O((m + n)^{4/3} / L^{1/3} \cdot \log (\frac{m + n}{L}))$ time. It remains to determine $|S_2'|$. 
A method is proposed in Lemma~4.2 of \cite{ref:AvrahamTh15} to determine with high probability whether $|S_2'|\leq L/2$. This is done by generating a random sample $R_2$ of $c_2 (M/L \cdot \log (m + n))$ values from $S_2$, for a sufficiently large constant $c_2 > 0$, and then check how many of them lie in $(\alpha, \beta]$.
The runtime of this step is $O(|R_2|)$, i.e., $O(M/L \cdot \log (m + n)) = O((m + n)^{4/3} / L^{1/3} \cdot \log (m + n))$.

If $|S_1| \leq L/2$ and the above approach determines that $|S'_2| \leq L/2$, then with high probability the total number of distances of $\calE(A,B)\cap (\alpha, \beta]$ is at most $L$ and we are done with the lemma. Otherwise, an approach is given in Lemma 4.3 of \cite{ref:AvrahamTh15} to generate a sample $R$ of $O(\log (m + n))$ distances from $S_1\cup S_2$, so that with high probability $R$ contains an approximate median (in the middle three quarters) among the values of $\calE(A,B)$ in $(\alpha, \beta]$; this step takes $O((m + n)^{4/3} / L^{1/3} \cdot \log (m + n))$ time. 

We now call the decision algorithms to do binary search on the values of $R$ to find two consecutive values $\alpha', \beta'$ in $R$ such that $\delta^* \in (\alpha', \beta']$. Note that $(\alpha', \beta']\subseteq (\alpha, \beta]$, and $(\alpha', \beta']$ contains with high probability at most $7/8$ distances of $\calE(A,B)$ in $(\alpha, \beta]$.
As $|R|=O(\log (m+n))$, we need to call the decision algorithm $O(\log\log (m+n))$ times, and thus computing $(\alpha',\beta']$ takes $O(T_D \cdot \log \log (m + n))$ time. This finishes one iteration of the algorithm, which takes $O((m + n)^{4/3} / L^{1/3} \cdot \log (m + n) + T_D \cdot \log \log (m + n))$ time in total.


We then proceed with the next iteration with $(\alpha,\beta]=(\alpha',\beta']$. The exptected number of iterations of the algorithm is $O(\log (m+n))$. Hence, the expected time of the overall algorithm is $O((m + n)^{4/3} / L^{1/3} \cdot \log^2 (m + n) + T_D\cdot \log (m + n) \cdot \log\log (m + n))$. 
\end{proof}


With the interval $(\alpha,\beta]$ computed by Lemma~\ref{theorem:LDistancesFramework}, the next step is to compute $\delta^*$ from $\calE(A,B)\cap (\alpha,\beta]$. This is done using bifurcation tree technique, which was initially proposed in 
\cite{ref:AvrahamTh15} but was made more general in \cite{ref:KaplanTh23}. Suppose the decision problem can be solved in $T^*_D$ time by an algorithm with the following {\em special property}: All critical values are interpoint distances (i.e., if we simulate the decision algorithm on $\delta^*$ without knowing $\delta^*$ like the standard parametric search~\cite{ref:MegiddoAp83}, then every value that needs to be compared to $\delta^*$ during the algorithm is required to be an interpoint distance; such a value is called a {\em critical value}).\footnote{Katz and Sharir~\cite{ref:KatzEf21arXiv} originally proposed a general algorithm framework but overlooked this issue that a special decision algorithm is required, as acknowledged in \cite{ref:KaplanTh23}.} 
Then, the runtime of this step is $O(T^*_D\cdot L^{1/2}\cdot \log^{1/2}(m+n))$ (see \cite[Section~3.3]{ref:KaplanTh23} by replacing $n$ with $n+m$). In fact, by a more careful analysis, we can bound the time by $O(\sqrt{L\cdot T_D\cdot T_D^*\cdot  \log(m+n)})$ (note that $T_D\leq T_D^*$ and thus this new bound is better than the previous one). We briefly explain this. For simplicity, we use the notation from \cite[Section~3.3]{ref:KaplanTh23} without explanation. The main observation is that we can still use an ordinary decision algorithm of $T_D$ time to do binary search on critical values during the algorithm. Hence, each phase of the algorithm costs $O(XY+C_0(T)+T_D\cdot \log n)$ time. The number of phases is at most $\max\{L/X,T^*_D/Y\}$. By setting $Y=\sqrt{T^*_D\cdot T_D\cdot \log n/L}$ and making $L/X=T^*_D/Y$, we have $XY=T_D\cdot \log n$, and thus the total time is bounded as stated above (by replacing $n$ with $n+m$). 


In summary, the total time of the algorithm is $O((m + n)^{4/3} / L^{1/3} \cdot \log^2 (m + n) + T_D\cdot \log (m + n) \cdot \log \log (m + n) + \sqrt{L\cdot T_D\cdot T_D^*\cdot  \log(m+n)})$. We thus have the following theorem. 

\begin{theorem}
    \label{theorem:framework2}
    Given two sets $A$ and $B$ of $m$ and $n$ points respectively in the plane, any geometric optimization problem whose optimal objective value is equal to the distance between a point of $a\in A$ and a point of $b\in B$ can be solved by a randomized algorithm of $O((m + n)^{4/3} / L^{1/3} \cdot \log^2 (m + n) + T_D\cdot \log (m + n) \cdot \log\log (m + n) + \sqrt{L\cdot T_D\cdot T_D^*\cdot  \log(m+n)})$ expected time, for any parameter $1\leq L\leq mn$.
\end{theorem}

For the one-sided DFD problem, we have both $T_D, T^*_D=O(m+n)$ as the decision algorithm in \cite{ref:AvrahamTh15} has the special property on the critical values. Setting $L = (m + n)^{2/5}\log^{9/5}(m+n)$ leads to the following result.

\begin{corollary}
    \label{theorem:OS-DFD}
    Given a sequence $A$ of $m$ points and another sequence $B$ of $n$ points in the plane, the one-sided discrete Fr\'{e}chet distance with shortcuts problem can be solved by a randomized algorithm of $O((m + n)^{6/5} \log^{7/5} (m + n))$ expected time.
\end{corollary}

\paragraph{Reverse shortest paths in unit-disk graphs.}
As discussed in Section~\ref{sec:Introduction}, another immediate application of Theorem~\ref{theorem:framework2} is the reverse shortest path problem in unit-disk graphs. 

For the unweighted case, we can have $T_D=O(n)$ (after points of $P$ are sorted)~\cite{ref:ChanAl16}. However, as indicated in \cite{ref:KaplanTh23}, the algorithm of \cite{ref:ChanAl16} does not have the special property. Instead, an $O(n\log n)$ time decision algorithm is derived in \cite{ref:KaplanTh23} (see its full version on arXiv) and the algorithm has the special property. As such, we have $T^*_D=O(n\log n)$. Hence, applying Theorem~\ref{theorem:framework2} (by replacing $m+n$ with $n$ and setting $L=n^{2/5} \log^{6/5}n$) can solve the unweighted case in $O(n^{6/5} \log^{8/5} n)$ expected time. 

For the weighted case, the decision problem is solvable in $O(n\log^2 n)$ time~\cite{ref:WangNe20}. As indicated in \cite{ref:KaplanTh23}, the algorithm of \cite{ref:WangNe20} does not have the special property. Instead, Kaplan et al.~\cite{ref:KaplanTh23} (see its full version on arXiv) modified the algorithm of \cite{ref:WangNe20} to make the special property holds and the time of the modified algorithm is still $O(n\log^2 n)$. As such, we have both $T_D,T^*_D=O(n\log^2n)$. Hence, applying Theorem~\ref{theorem:framework2} (by replacing $m+n$ with $n$ and setting $L=n^{2/5}/\log^{3/5}n$) can solve the weighted case in $O(n^{6/5} \log^{11/5} n)$ expected time. 

\footnotesize
\bibliographystyle{plain}
\bibliography{reference}
\end{document}